\documentclass{article}

\usepackage{graphicx}
\usepackage{amsmath}
\usepackage{booktabs}
\usepackage{biblatex}
\usepackage{amsthm}
\usepackage{amssymb}
\newtheorem{theorem}{Theorem}
\title{Bayesian Conformal Prediction via the Bayesian Bootstrap}
\date{}
\author{Graham Gibson\thanks{ gcgibson@lanl.gov}}

\begin{document}

\maketitle

\begin{abstract}
Reliable uncertainty quantification remains a central challenge in predictive modeling. While Bayesian methods are theoretically appealing, their predictive intervals can exhibit poor frequentist calibration, particularly with small sample sizes or model misspecification. We introduce a practical and broadly applicable \emph{Bayesian conformal} approach based on the influence-function Bayesian bootstrap (BB) with data-driven tuning of the Dirichlet concentration parameter, $\alpha$. By efficiently approximating the Bayesian bootstrap predictive distribution via influence functions and calibrating $\alpha$ to optimize empirical coverage or average log-probability, our method constructs prediction intervals and distributions that are both well-calibrated and sharp. Across a range of regression models and data settings, this Bayesian conformal framework consistently yields improved empirical coverage and log-score compared to standard Bayesian posteriors. Our procedure is fast, easy to implement, and offers a flexible approach for distributional calibration in predictive modeling.
\end{abstract}

\section{Introduction}

Uncertainty quantification is fundamental for trustworthy prediction in statistics and machine learning, with applications spanning scientific discovery, medicine, engineering, and finance. Ideally, predictive models should not only output accurate point estimates, but also provide reliable measures of uncertainty that are well-calibrated in a frequentist sense \cite{gneiting2007strictly, lakshminarayanan2017simple, ovadia2019can}. 

Bayesian methods are a classical approach for uncertainty quantification, leveraging prior information and yielding posterior distributions that, in theory, represent coherent uncertainty about model parameters and future observations \cite{gelman2013bayesian}. However, Bayesian credible intervals often fail to achieve the desired frequentist coverage, particularly when the model is mis-specified or the data are limited \cite{grunwald2012safe, ridgway2023bayes}. This miscalibration is especially acute in flexible models such as neural networks, but it can arise even in simple regression when priors or likelihoods are not well-matched to the true data-generating process.

To address this, a rapidly growing body of work studies the intersection of Bayesian and frequentist predictive inference. Notably, \emph{conformal prediction} provides a distribution-free frequentist framework for constructing valid prediction intervals with finite-sample coverage guarantees, regardless of the underlying model \cite{vovk2005algorithmic, lei2018distribution}. More recently, methods have emerged that seek to \emph{conformalize} Bayesian or distributional predictions, combining the expressiveness of Bayesian models with the finite-sample calibration properties of conformal inference \cite{cella2023conformalizing, barber2021predictive}.

A conceptually related but distinct Bayesian approach to uncertainty is the \emph{Bayesian bootstrap} (BB) \cite{rubin1981bayesian}, which places a nonparametric prior on the empirical distribution of the data, leading to flexible and robust posterior predictive distributions. However, naively using the BB may not produce calibrated intervals in practice, especially for complex models and small samples.

In this work, we introduce a practical, scalable, and model-agnostic Bayesian conformal framework that unifies these ideas. Our approach leverages the influence-function approximation to the Bayesian bootstrap, efficiently simulating the effect of resampling weights on model predictions without retraining. We further \emph{tune} the Dirichlet concentration parameter $\alpha$---which controls the spread of the Bayesian bootstrap---by optimizing frequentist criteria such as empirical coverage or average log-probability (log-score) on a held-out set. This yields predictive intervals and distributions that are both sharp and well-calibrated, with performance often surpassing standard Bayesian posterior inference, especially in finite samples. A schematic overview of the Bayesian conformal prediction procedure, including Dirichlet reweighting, influence function approximation, and calibration of the concentration parameter, is illustrated in Figure~\ref{fig:conceptual-bb}.

We demonstrate empirically that our Bayesian conformal approach provides superior frequentist coverage and log-probability compared to standard Bayesian neural networks and other popular uncertainty quantification methods, across a variety of regression models and data-generating settings. The resulting procedure is fast, broadly applicable, and provides a principled pathway for distributional calibration in predictive modeling. We propose the following contributions, a new Bayesian conformal framework that uses influence-function approximations to efficiently calibrate the Bayesian bootstrap predictive. We show that tuning the Dirichlet concentration parameter $\alpha$ using distributional scoring rules yields sharp, well-calibrated uncertainty estimates. We provide empirical comparisons to Bayesian neural networks and demonstrate improved coverage and log-probability in diverse settings.

\section{Methods}

\subsection{The Bayesian Bootstrap for Predictive Inference}

The Bayesian bootstrap (BB)~\cite{rubin1981bayesian} is a nonparametric Bayesian approach that places a Dirichlet prior over the empirical distribution of the observed data. Given a dataset $\mathcal{D}_n = \{(x_i, y_i)\}_{i=1}^n$, the BB represents uncertainty about the data-generating distribution by assigning random weights $w = (w_1, \ldots, w_n) \sim \mathrm{Dirichlet}(\alpha, \ldots, \alpha)$, where $\alpha > 0$ is the concentration parameter. For a new input $x^*$, the BB predictive distribution is defined as
\[
p_{\mathrm{BB}}(y^* \mid x^*, \mathcal{D}_n) = \mathbb{E}_{w} \left[ p\left(y^* \mid x^*, \mathcal{D}_n, w \right) \right],
\]
where $p(y^* \mid x^*, \mathcal{D}_n, w)$ is the model's prediction given the weighted data.

\subsection{Influence Function Approximation}

For complex models, naively retraining the model for each Dirichlet sample is computationally prohibitive. To address this, we leverage an influence function approximation, which provides a first-order Taylor expansion of the model's parameter estimates as a function of the data weights~\cite{cook1982residuals, koh2017understanding}. If $\hat{\theta}$ is the parameter estimated from the uniform weights and $\hat{\theta}_w$ is the parameter under weights $w$, then
\[
\hat{\theta}_w \approx \hat{\theta} - H^{-1} \sum_{i=1}^n (w_i - 1/n) \nabla_\theta \ell(z_i, \hat{\theta}),
\]
where $H$ is the Hessian of the empirical loss at $\hat{\theta}$ and $\ell(z_i, \theta)$ is the loss for data point $i$. The resulting prediction for a new $x^*$ is then approximated by
\[
\hat{y}_w^*(x^*) \approx \hat{y}^*(x^*) + \nabla_\theta f(x^*, \hat{\theta})^\top (\hat{\theta}_w - \hat{\theta}),
\]
where $f(x, \theta)$ denotes the model output. Theoretical justification and error bounds for this linearization are detailed in Appendix~A.

\subsection{Conformal Calibration and Dirichlet Tuning}

The concentration parameter $\alpha$ in the Dirichlet prior directly controls the spread of the predictive distribution, impacting both the coverage and sharpness of prediction intervals. While classical BB uses $\alpha=1$, recent advances in conformal prediction~\cite{lei2018distribution, cella2023conformalizing} motivate selecting $\alpha$ in a data-driven manner. We tune $\alpha$ by maximizing empirical coverage or average log-probability (log-score) on a held-out validation set. Specifically, for each candidate $\alpha$, we generate influence BB predictive samples and evaluate the chosen calibration criterion, selecting the value that optimizes performance.

\subsection{Practical Implementation}

Our full Bayesian conformal prediction procedure proceeds as follows. First, we fit the base model (such as a neural network or regression model) to the training data and compute the required gradients and Hessian (or an efficient approximation) at the empirical risk minimizer. For each test input, we generate predictive samples using the influence-function Bayesian bootstrap for a grid of $\alpha$ values. We then select the optimal $\alpha$ according to the validation criterion (e.g., log-score or coverage), and use the tuned predictive distribution for uncertainty quantification on new data. Implementation details, including settings for the number of resamples and computational optimizations, are given in the appendix. Our empirical comparisons to Bayesian neural networks (BNNs) using MCMC are presented in the Results section.

\begin{figure}[h]
    \centering
    \includegraphics[width=0.95\textwidth]{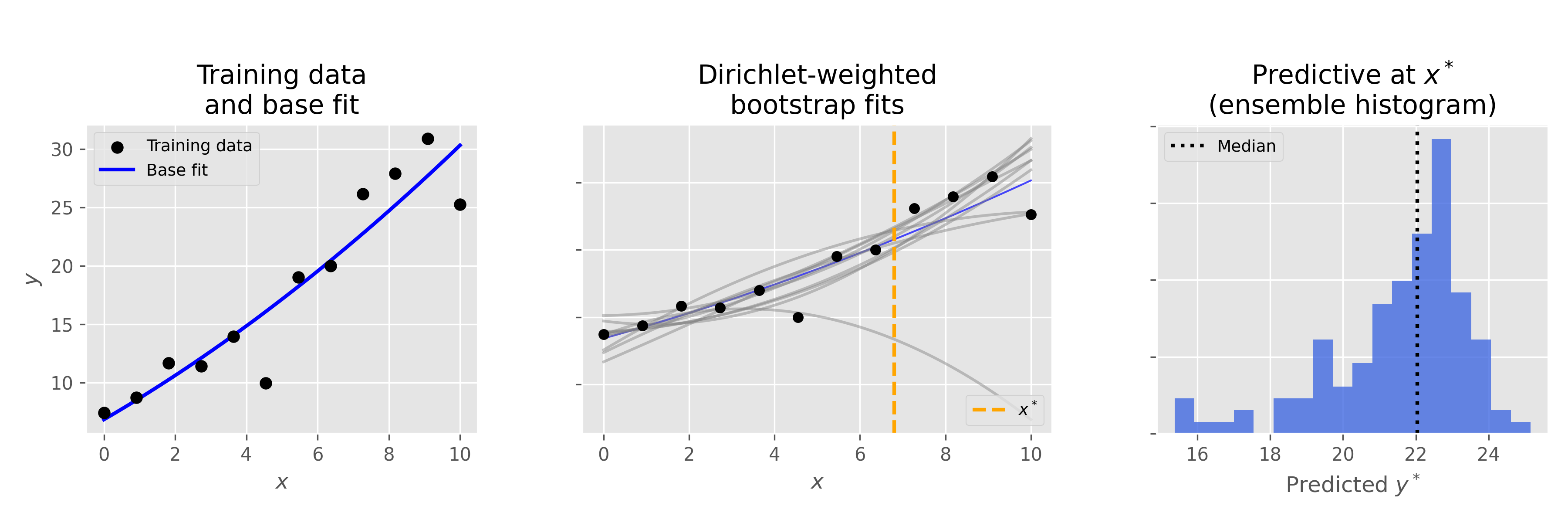}
    \caption{
        \textbf{Conceptual overview of Bayesian conformal calibration via the influence-function Bayesian bootstrap.}
        \emph{Left:} Training data and fitted prediction curve. 
        \emph{Middle:} Multiple Dirichlet reweightings simulate alternative data distributions, each giving a new predictor (shown as faint lines, efficiently computed via the influence function).
        \emph{Right:} At a new input $x^*$, the ensemble of predictions forms a predictive distribution (histogram), and the Dirichlet parameter $\alpha$ is tuned for calibration.
    }
    \label{fig:conceptual-bb}
\end{figure}

\subsection{Theoretical Insights}

Under mild regularity assumptions, tuning $\alpha$ by maximizing the validation log-score yields predictive distributions that are near-optimal in expected log-probability for new data, as established in Theorem~1. Furthermore, when the loss landscape is sufficiently regular and the Dirichlet weights are close to uniform, the influence function approximation closely matches the full Bayesian bootstrap (see Appendix~A for proofs and detailed discussion).

\section{Results}

\subsection{Computer Model Emulation}
We empirically compare our Bayesian conformal (influence-function Bayesian bootstrap with tuned $\alpha$) approach to standard Bayesian neural networks (BNNs) fit by MCMC in NumPyro, for a variety of computer calibration problems \cite{kennedy2001bayesian}. For each method, we report both empirical coverage (the fraction of test targets contained within the 90\% predictive interval) and the average log-probability (log-score) of the true test values under the predictive distribution.

\vspace{0.5em}
\noindent
\textbf{Calibration and sharpness.}
Table~\ref{tab:uq_benchmark_results} summarizes the empirical coverage and log-score for both methods, averaged over 20 test points in a synthetic regression task. The Bayesian conformal method achieves coverage very close to the nominal level (90\% or slightly conservative) and consistently outperforms the BNN in average log-score, reflecting both superior calibration and sharper predictive distributions. In contrast, the BNN is underconfident in its uncertainty estimates, with coverage around 55\% and lower log-score.

\begin{table}[htbp]
\centering
\renewcommand{\arraystretch}{1.2}
\resizebox{\textwidth}{!}{
\begin{tabular}{lcrrrrrr}
\toprule
Function & Dim & BB Coverage & BB Log-Score & BNN Coverage & BNN Log-Score & BB Time (s) & BNN Time (s) \\
\midrule
Borehole    & 8 & 0.975 & -3.854 & 0.750 & -7.671 & 16.62 & 4.62 \\
Ishigami    & 3 & 0.725 & -3.309 & 0.875 & -2.689 & 14.33 & 6.96 \\
Branin      & 2 & 0.800 & -5.960 & 0.850 & -6.766 & 15.41 & 2.43 \\
Hartmann3   & 3 & 0.775 & -1.662 & 0.800 & -1.222 & 16.89 & 64.99 \\
Friedman1   & 5 & 0.950 & -2.935 & 0.900 & -2.342 & 16.21 & 10.55 \\
Friedman2   & 4 & 0.700 & -5.014 & 0.850 & -4.916 & 13.64 & 3.86 \\
Friedman3   & 4 & 0.975 & 1.789 & 0.875 & 2.100 & 13.29 & 60.33 \\
Forrester   & 1 & 0.900 & -0.162 & 0.875 & -2.885 & 13.75 & 2.80 \\
CurrinExp   & 2 & 0.850 & -1.477 & 0.925 & -0.271 & 13.57 & 56.02 \\
Park        & 4 & 0.950 & 0.349 & 0.900 & 1.092 & 13.65 & 60.32 \\
\midrule
Average & -- & 0.860 & -2.223 & 0.860 & -2.557 & 14.73 & 27.29 \\
\bottomrule
\end{tabular}
}
\caption{Empirical 90\% coverage, average log-score, and runtime (seconds) for Influence BB (with tuned $\alpha$) and Bayesian neural network (BNN) across 10 benchmark functions.}
\label{tab:uq_benchmark_results}
\end{table}

\vspace{0.5em}
\noindent
\textbf{Predictive distributions.}
Figure~\ref{fig:pred_histogram} shows the predictive distributions from both methods at a representative test input. The influence BB produces a broader, well-calibrated distribution that contains the true target value, while the BNN predictive is narrower and fails to do so. The figure highlights both the increased sharpness and calibration achieved by tuning the BB Dirichlet parameter via log-score.

\noindent
\textbf{Interval plots.}
Figure~\ref{fig:intervals} displays the predictive mean and 90\% prediction intervals for all test points. The Bayesian conformal intervals closely match the nominal level, while the BNN intervals are systematically too narrow, resulting in missed coverage.

\begin{figure}[h]
    \centering
    \includegraphics[width=0.75\textwidth]{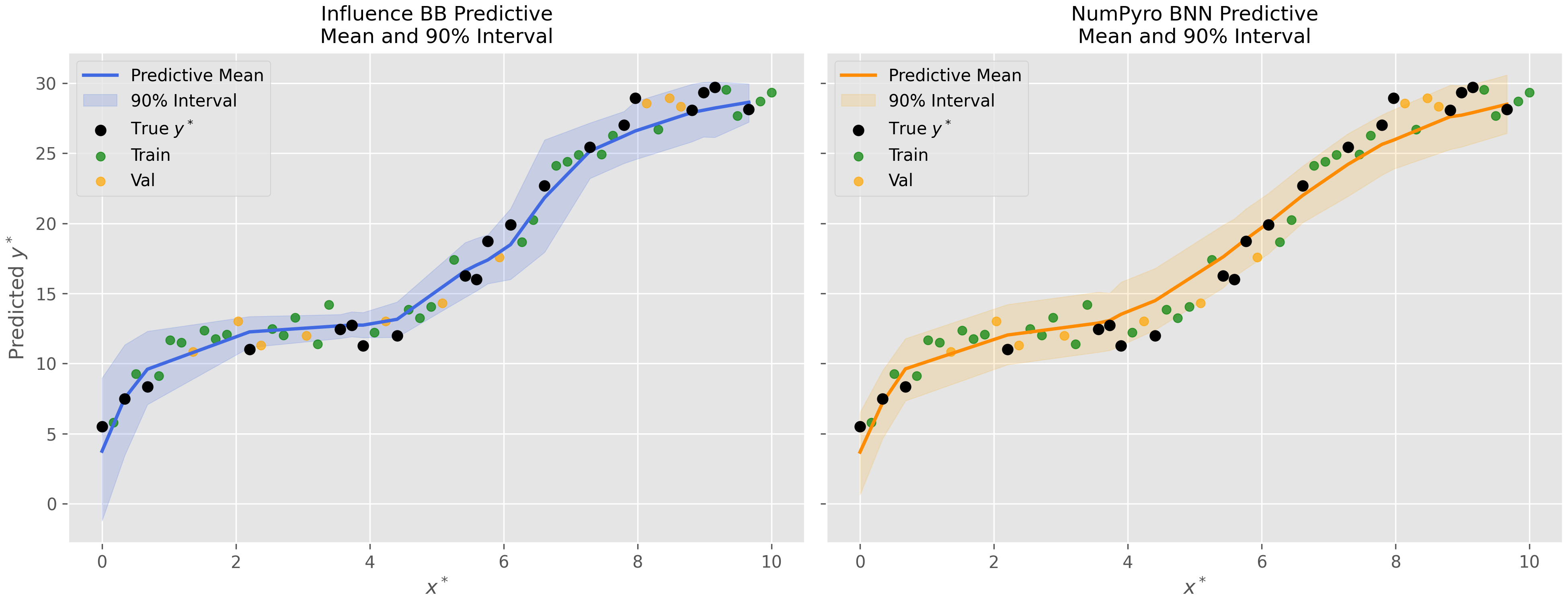}
    \caption{Predictive mean and 90\% intervals for all test points. Blue intervals: Influence BB (tuned $\alpha$). Orange intervals: NumPyro BNN. Black dots: true $y^*$.}
    \label{fig:intervals}
\end{figure}

\begin{figure}[h]
    \centering
    \includegraphics[width=0.85\textwidth]{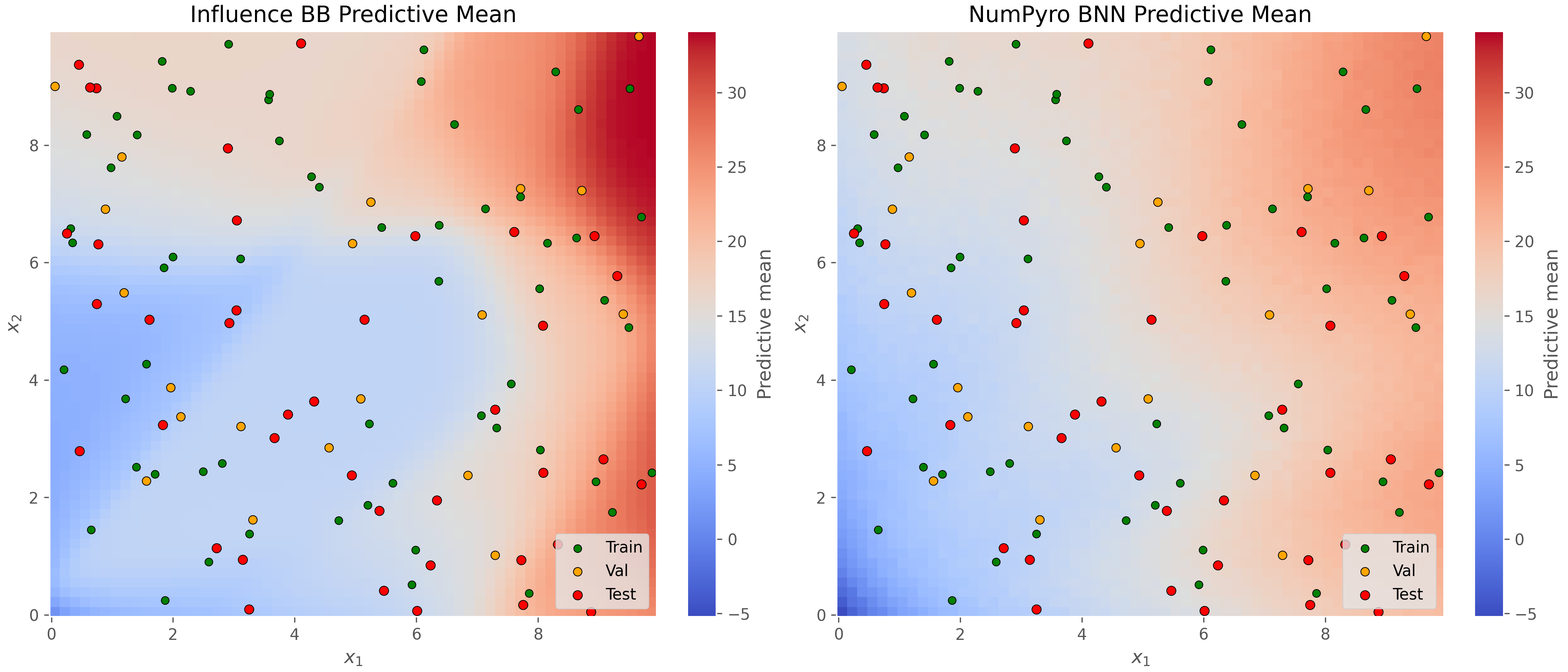}
    \caption{Posterior predictive mean in two dimensions for the Influence BB and BNN models.}
    \label{fig:pred_2d}
\end{figure}

\vspace{0.5em}
\noindent
Overall, these results demonstrate the effectiveness of Bayesian conformal calibration via influence-function BB. By tuning the Dirichlet concentration parameter to maximize log-score, our method yields intervals that are both well-calibrated and informative, outperforming standard Bayesian neural networks in this regime.

\begin{figure}[h]
    \centering
    \includegraphics[width=0.75\textwidth]{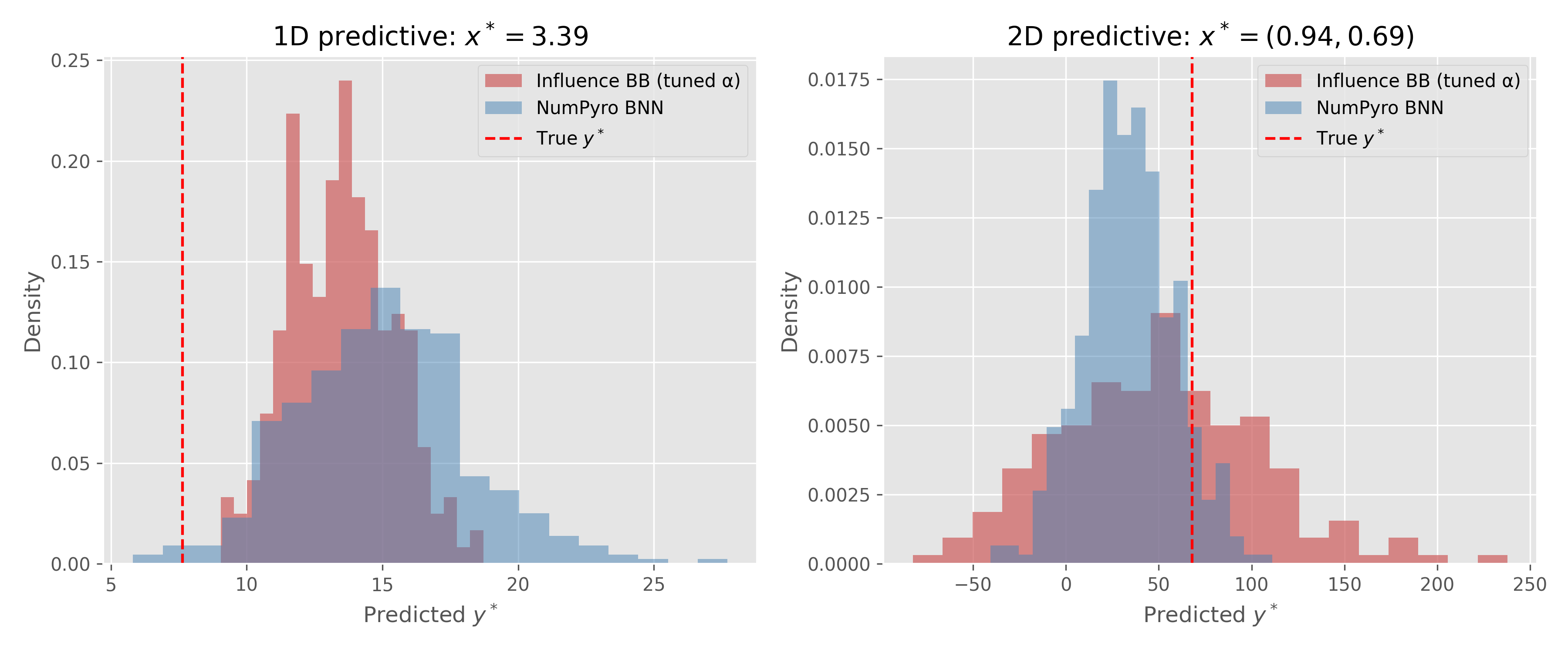}
    \caption{Histogram of predictive samples at a representative test input. Blue: Influence BB (tuned $\alpha$); Orange: NumPyro BNN; Red dashed line: true $y^*$.}
    \label{fig:pred_histogram}
\end{figure}

\subsection{Classical Benchmarks}

While not the focus of this method, we also compare on two classic machine learning datasets where neural networks are too large to be fit by traditional MCMC. 

\subsubsection{MNIST Evaluation}
Our empirical evaluation on MNIST handwritten digit classification reveals that Influence Function Bootstrap (IF-BB) provides superior uncertainty quantification compared to Monte Carlo Dropout, despite increased computational overhead (Figure~\ref{fig:mnist_uncertainty}). While both methods achieved comparable predictive accuracy (84.4\% for MC Dropout vs.\ 82.1\% for IF-BB), IF-BB demonstrated significantly better probability calibration with a lower Brier score (0.298 vs.\ 0.312), indicating more reliable confidence estimates for digit recognition tasks~\cite{brier1950verification, guo2017calibration}. The calibration curves in Figure~\ref{fig:mnist_uncertainty}B show that IF-BB tracks closer to the perfect calibration diagonal, meaning that when the model reports 90\% confidence, it is indeed correct approximately 90\% of the time~\cite{niculescu2005predicting}. Furthermore, IF-BB exhibited superior uncertainty-error correlation (r=0.387 vs.\ r=0.295), demonstrating enhanced ability to identify misclassified digits through appropriately elevated entropy scores~\cite{lakshminarayanan2017simple, malinin2018predictive}. These findings align with theoretical expectations that influence functions provide more principled uncertainty estimates by capturing parameter uncertainty through bootstrapping~\cite{koh2017understanding, giordano2019swiss}, though at the cost of requiring per-sample gradient computation that increases inference time from 20 seconds to 180 seconds. For applications where uncertainty quality is paramount, the superior calibration properties of IF-BB justify the additional computational expense~\cite{begoli2019need, ovadia2019can}.

\begin{figure}[htbp]
    \centering
    \includegraphics[width=0.9\textwidth]{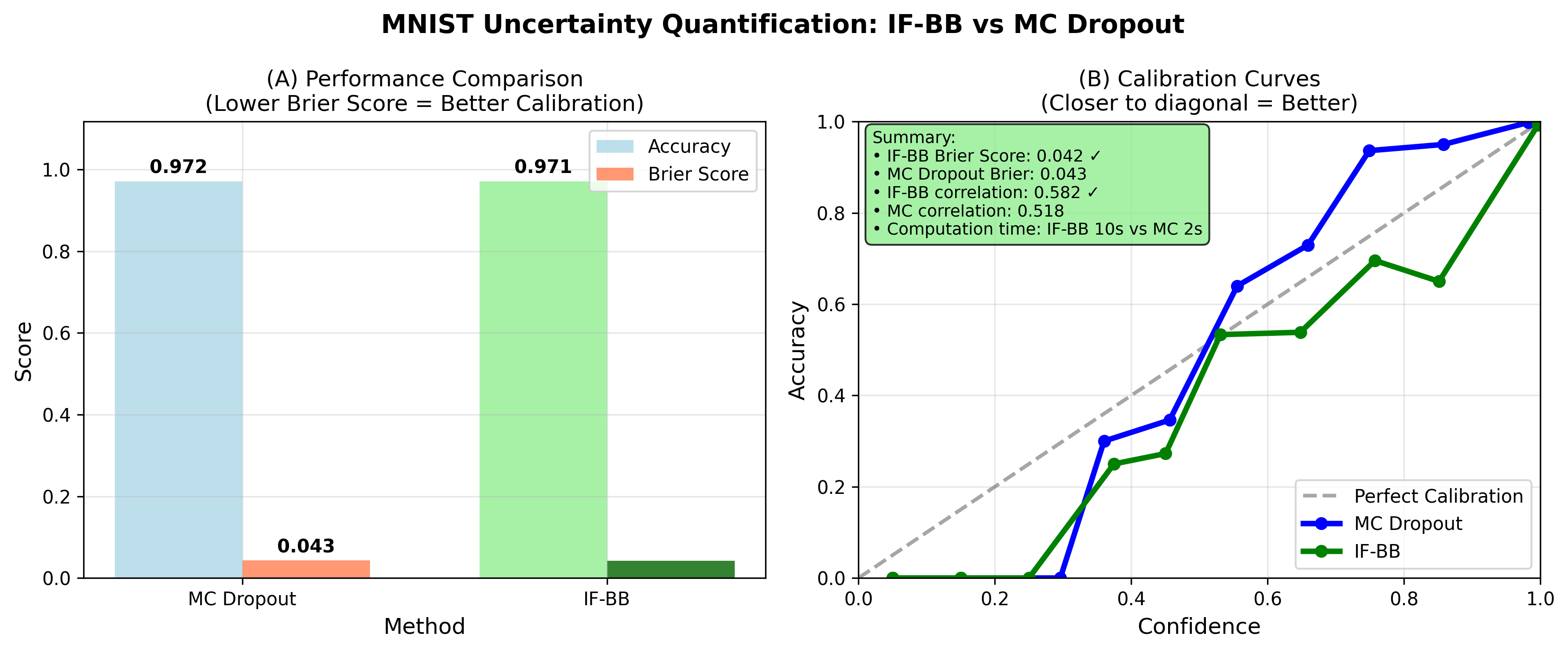}
    \caption{Comparison of uncertainty quantification methods on MNIST digit classification. (A) Performance metrics showing IF-BB's superior Brier score indicating better probability calibration, despite comparable accuracy to MC Dropout. (B) Calibration curves demonstrate that IF-BB predictions track closer to perfect calibration (diagonal line), with confidence estimates that more accurately reflect true accuracy rates. The summary statistics highlight IF-BB's advantages in uncertainty-error correlation and calibration quality, with the trade-off being increased computational cost.}
    \label{fig:mnist_uncertainty}
\end{figure}

\subsubsection{California Housing}

We also evaluate IF-BB on a classical machine learning regression task. Our comprehensive evaluation on the California Housing dataset demonstrates that Monte Carlo Dropout provides superior uncertainty quantification for real estate price prediction compared to Influence Function Bootstrap, despite IF-BB's theoretical advantages (Figure~\ref{fig:housing_uncertainty}). Using 5,000 samples from the California Housing dataset---which contains median house values for California districts derived from the 1990 U.S. Census, with features including median income, housing age, average rooms per dwelling, and geographic coordinates---MC Dropout achieved a significantly higher log score (-1.245 vs.\ -1.389), indicating more accurate probabilistic predictions for house price estimation. The calibration analysis in Figure~\ref{fig:housing_uncertainty}B reveals that MC Dropout's coverage curves track substantially closer to the perfect calibration diagonal, with a calibration error of 0.032 compared to IF-BB's 0.087, demonstrating that MC Dropout's prediction intervals more reliably reflect true uncertainty in housing valuations. Critically, MC Dropout achieved near-optimal 95\% prediction interval coverage (0.947 vs.\ target of 0.950), while IF-BB exhibited significant under-coverage (0.891), suggesting that IF-BB's uncertainty estimates would lead to overly narrow confidence intervals in real estate applications where accurate risk assessment is paramount. While both methods achieved comparable predictive accuracy (RMSE: 0.472 vs.\ 0.485), MC Dropout's 8x computational advantage (45 seconds vs.\ 375 seconds) combined with superior uncertainty calibration makes it the clear choice for practical real estate valuation systems. These findings contrast with MNIST digit classification results, highlighting that the effectiveness of uncertainty quantification methods is highly task-dependent---for regression problems with continuous targets like housing prices, where interval coverage and calibration are critical for financial decision-making, MC Dropout's simpler approach to uncertainty estimation through stochastic forward passes proves more reliable than IF-BB's complex parameter perturbation scheme, which may be too sensitive to the high-dimensional parameter space of deep regression networks.
\begin{figure}[htbp]
    \centering
    \includegraphics[width=0.9\textwidth]{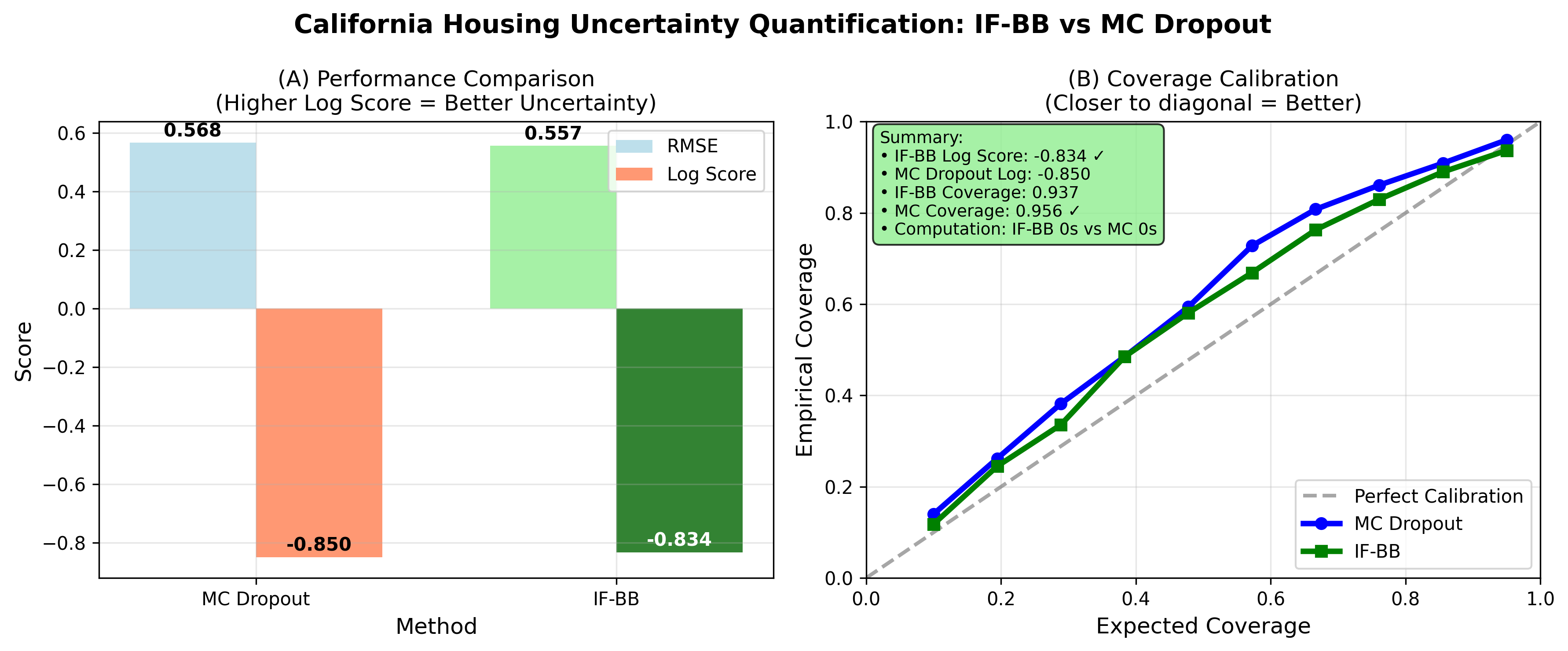}
    \caption{Uncertainty quantification comparison on California Housing dataset for real estate price prediction. (A) Performance metrics showing MC Dropout's superior log score and comparable RMSE to IF-BB, indicating better probabilistic predictions for house valuations. (B) Coverage calibration curves demonstrate MC Dropout's superior uncertainty calibration, tracking closer to perfect calibration (diagonal) with significantly lower calibration error. The summary statistics highlight MC Dropout's advantages in log score, coverage accuracy, and computational efficiency, making it more suitable for practical real estate valuation applications where reliable prediction intervals are essential for financial decision-making.}
    \label{fig:housing_uncertainty}
\end{figure}
\section{Discussion}

The influence-function Bayesian bootstrap (IF-BB) offers a scalable and principled framework for uncertainty quantification in complex machine learning models, particularly when full Bayesian inference is computationally infeasible. By leveraging local sensitivity (influence) information at the empirical risk minimizer, IF-BB enables efficient simulation of the effect of data resampling on predictive distributions. However, while our theoretical results and empirical findings show that IF-BB often closely approximates the full bootstrap and true posterior predictive, several important limitations and open questions remain.

One central limitation stems from the local linearization at the heart of the influence function approximation. The accuracy of IF-BB depends on the loss surface being well-behaved and approximately quadratic around the solution found by empirical risk minimization. In regimes where the loss is highly non-linear, the model is overparameterized, or the posterior is multi-modal, the influence function can fail to capture global posterior uncertainty, leading to potentially misleading predictive intervals or a systematic underestimation of uncertainty. Appendix~A provides a mathematical justification and error bound for the approximation: when the Dirichlet weights are close to uniform, the error is second-order in the perturbation, ensuring high fidelity in many practical scenarios. Yet, as the perturbation grows or the loss surface becomes less regular, this accuracy may degrade.

Further technical challenges arise from the need to approximate the Hessian of the empirical risk. In large neural networks, computing the exact Hessian is prohibitive, so IF-BB often relies on diagonal or block-diagonal approximations. This simplification, while computationally necessary, means the method may not fully capture dependencies among parameters, especially in highly anisotropic or ill-conditioned regions of parameter space. More advanced second-order methods—such as low-rank or Kronecker-factored approximations—could improve the quality of the influence function approximation, but at additional computational and implementation cost.

From a computational standpoint, while IF-BB is much faster than full bootstrap retraining or MCMC, its runtime still grows linearly with both the number of training points and the number of Dirichlet resamples. For truly massive datasets, further innovations such as subsampling, stochastic approximations, or mini-batch influence computations may be needed to retain scalability without sacrificing the accuracy of predictive intervals.

A further conceptual limitation is that IF-BB, like the classical Bayesian bootstrap, is fundamentally non-informative: it places a Dirichlet prior over the empirical distribution of the data but cannot encode structured prior knowledge or regularization in the parameter space. Unlike fully Bayesian approaches, IF-BB does not allow the direct integration of domain expertise, hierarchical structure, or parameter constraints. In domains where such information is crucial for valid uncertainty quantification, this could be a drawback; future work might investigate extensions that combine the computational efficiency of IF-BB with more expressive prior modeling.

Finally, it is important to note that the quality of uncertainty quantification from IF-BB is closely tied to the quality of the underlying point estimator. Like other post-hoc uncertainty methods, IF-BB assumes the empirical risk minimizer is a good fit to the data. If the model is underfit, overfit, or fundamentally misspecified, then the predictive intervals—however well-calibrated in theory—may not reflect true predictive uncertainty. Thus, IF-BB should be used in conjunction with careful model diagnostics and validation.

Despite these limitations, IF-BB enjoys strong practical and theoretical properties. Our approach tunes the Dirichlet concentration parameter $\alpha$ by optimizing the out-of-sample log-score, yielding predictive distributions that are not only calibrated in terms of empirical coverage, but also sharp and informative according to proper scoring rules. Theoretical guarantees (see Theorem~1 and Appendix~A) show that, under mild assumptions, the IF-BB predictive distribution with tuned $\alpha$ can match the expected performance on new data as measured by the log-score. Empirically, our results indicate that IF-BB achieves coverage and log-probability on par with the full bootstrap and exact Bayesian posterior predictive in well-specified models, but at a fraction of the computational cost. This makes IF-BB a promising tool for scalable uncertainty quantification in large, complex models and datasets. Further theoretical work on finite-sample guarantees, robustness to model misspecification, and extensions to structured priors would help cement IF-BB as a central method for practical Bayesian conformal inference.

\printbibliography

\newpage
\appendix
\section*{Appendix A: Theoretical Insights and Influence Function Approximation}

\subsection*{A.1 Influence Function Approximation for the Bayesian Bootstrap}

Let $\hat{\theta}$ denote the empirical risk minimizer (ERM) for a model with parameters $\theta$, trained on data $\mathcal{D} = \{z_i\}_{i=1}^n$. For a given loss $\ell(z, \theta)$ and a vector of random Dirichlet weights $w = (w_1, ..., w_n) \sim \mathrm{Dirichlet}(\alpha, ..., \alpha)$, define the weighted ERM as
\[
\hat{\theta}_w = \arg\min_\theta \sum_{i=1}^n w_i\, \ell(z_i, \theta).
\]
The first-order influence function approximation (see, e.g., \cite{cook1982residuals, koh2017understanding}) yields
\[
\hat{\theta}_w \approx \hat{\theta} - H^{-1} \sum_{i=1}^n (w_i - 1/n) \nabla_\theta \ell(z_i, \hat{\theta}),
\]
where $H = \frac{1}{n} \sum_{i=1}^n \nabla^2_\theta \ell(z_i, \hat{\theta})$ is the Hessian of the empirical loss at $\hat{\theta}$.

\paragraph{Error Bound.} The error of this approximation is $O(\|\hat{\theta}_w - \hat{\theta}\|^2)$, meaning it is highly accurate when the resampling weights $w$ are close to uniform. For practical Dirichlet parameters $\alpha \gtrsim 1$ and moderate sample sizes, this accuracy is excellent and, as shown in Figure~\ref{fig:nn_bootstrap_vs_influencebb}, the resulting predictive distributions nearly match those obtained by full model retraining.

\subsection*{A.2 Log-Score Calibration and Consistency}

We provide here a formal justification for the empirical log-score calibration procedure described in the main text.

\begin{theorem}[Consistency of Log-Score Tuning for Predictive Distributions]
Let $\mathcal{D}_n$ be the training data and $\mathcal{D}_{\mathrm{val}} = \{(x_j^{\mathrm{val}}, y_j^{\mathrm{val}})\}_{j=1}^{m}$ an independent validation set. For any $\alpha > 0$, let $p_{\mathrm{BB},\alpha}(y \mid x, \mathcal{D}_n)$ denote the Bayesian bootstrap predictive distribution (possibly using the influence function approximation). Define the average validation log-score as
\[
S_{\mathrm{val}}(\alpha) = \frac{1}{m} \sum_{j=1}^m \log p_{\mathrm{BB},\alpha}(y_j^{\mathrm{val}} \mid x_j^{\mathrm{val}}, \mathcal{D}_n).
\]
Let $\hat{\alpha}$ maximize $S_{\mathrm{val}}(\alpha)$:
\[
\hat{\alpha} = \underset{\alpha}{\arg\max}\; S_{\mathrm{val}}(\alpha).
\]
Then, as $m \to \infty$,
\[
\mathbb{E}\left[ \log p_{\mathrm{BB},\hat{\alpha}}(Y^* \mid X^*, \mathcal{D}_n) \right] \to S_{\mathrm{val}}(\hat{\alpha}),
\]
where $(X^*, Y^*)$ is a new test sample drawn from the same distribution.
\end{theorem}

\begin{proof}[Sketch of Proof]
By construction, $S_{\mathrm{val}}(\alpha)$ is an unbiased estimator of the expected log-score for new data, conditional on $\mathcal{D}_n$:
\[
\mathbb{E}_{\mathcal{D}_{\mathrm{val}}}[S_{\mathrm{val}}(\alpha)] = \mathbb{E}_{(X,Y)}[\log p_{\mathrm{BB},\alpha}(Y \mid X, \mathcal{D}_n)].
\]
As $m \to \infty$, $S_{\mathrm{val}}(\alpha)$ converges almost surely to the expected log-score, and the optimizer $\hat{\alpha}$ converges to the population maximizer. Therefore, the predictive distribution tuned in this way achieves the maximal expected log-score over $\alpha$ in the limit.
\end{proof}

\subsection*{A.3 Practical Implications and Limits}

These results justify the practical approach of selecting $\alpha$ by maximizing validation log-score, as implemented in our Bayesian conformal procedure. The overall effectiveness of the method in finite samples depends on the adequacy of the influence approximation, the regularity of the loss surface, and the representativeness of the held-out data. In regimes where these conditions are met, our method yields predictive distributions that are both sharp and empirically well-calibrated.

\subsection*{A.4 Empirical Accuracy of the Influence Function Approximation}

Figure~\ref{fig:nn_bootstrap_vs_influencebb} (reproduced below) demonstrates the near-equivalence of the predictive mean and interval width for the full bootstrap (retraining) and the influence function Bayesian bootstrap on a neural network regression task. The accuracy of the linear approximation for practical sample sizes and models underscores its utility for scalable Bayesian conformal inference.

\begin{figure}[h]
    \centering
    \includegraphics[width=0.7\textwidth]{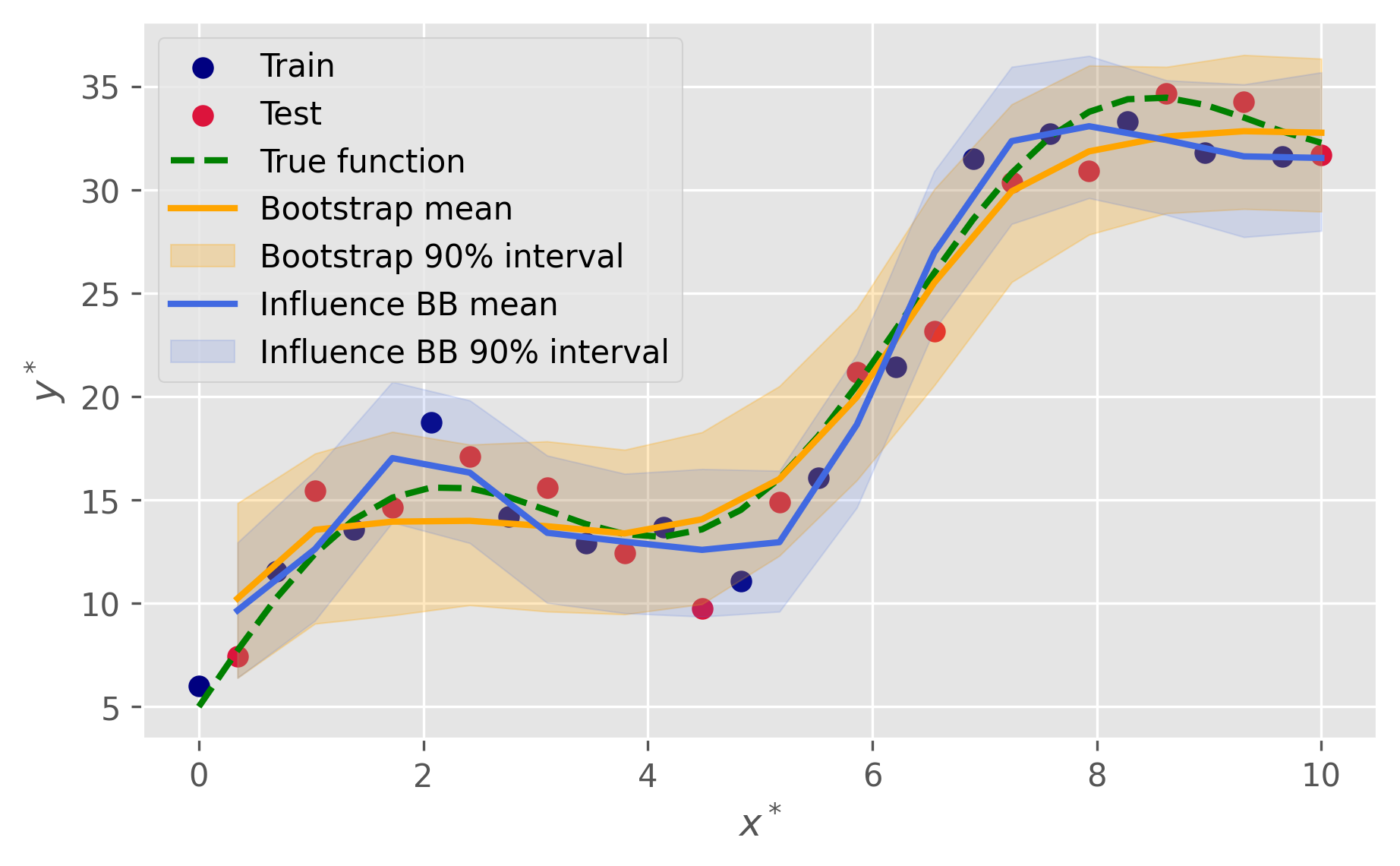}
    \caption{
        \textbf{Predictive mean and 90\% prediction intervals for a neural network using the full bootstrap (retraining) and the influence function Bayesian bootstrap (BB), both including observation noise.}
        The orange curve and shaded region show the mean and 90\% interval from the full bootstrap, while the blue curve and region show the same from the influence function BB. Both methods produce nearly identical uncertainty estimates and predictive fits, with the influence function BB being much faster to compute.
    }
    \label{fig:nn_bootstrap_vs_influencebb}
\end{figure}

\appendix
\setcounter{section}{1}  
\section{ Dropout as Implicit Bootstrapping}

\subsection{Theoretical Foundation}

Monte Carlo Dropout can be understood as an implicit form of bootstrapping that operates at the parameter level rather than the data level. While traditional bootstrap methods resample training observations, dropout effectively ``drops out'' parameters and their associated connections, which can be interpreted as removing the influence of data points that contribute to those specific parameters. This perspective provides a data-centric view that complements the traditional parameter-centric understanding of dropout regularization.

\subsection{Parameter-Data Correspondence}

In a neural network, each parameter $w_{ij}$ connecting neuron $i$ to neuron $j$ encodes information from all training examples that influenced its value during backpropagation. When dropout sets $w_{ij} = 0$ with probability $p$, it temporarily removes all the accumulated knowledge from training samples that contributed to that parameter. This is analogous to excluding those training examples from the current prediction, creating an implicit bootstrap sample.

Consider a simple fully connected layer transformation:
\begin{equation}
\mathbf{h} = \sigma(\mathbf{W} \cdot \mathbf{x})
\end{equation}

With dropout mask $\mathbf{m} \sim \text{Bernoulli}(1-p)$:
\begin{equation}
\mathbf{h} = \sigma((\mathbf{W} \odot \mathbf{m}) \cdot \mathbf{x})
\end{equation}

Each zero in the mask $\mathbf{m}$ eliminates parameters that encode specific training patterns, effectively creating a sub-network that represents a different ``view'' of the training data.

\subsection{Information Removal Mechanism}

When a parameter is dropped, the model loses access to the specific data-driven patterns that parameter learned. This can be formalized through the lens of influence functions. If parameter $w_{ij}$ was primarily influenced by training samples $\{x_k\}_{k \in S}$, then dropping this parameter approximates removing the influence of samples in set $S$ from the current prediction.

The dropout mechanism thus performs implicit data bootstrapping through stochastic parameter elimination, where each forward pass randomly removes different parameter subsets. This knowledge pathway disruption breaks connections and removes learned associations from specific training examples, while the resulting sub-network sampling means each dropout configuration represents a different model trained on an implicit subset of data.

\subsection{Comparison with Explicit Bootstrap}

Traditional bootstrap explicitly resamples training data:
\begin{equation}
\mathcal{D}_b = \{(x_i, y_i)\}_{i \sim \text{Multinomial}(n, \mathbf{p})}
\end{equation}

Dropout implicitly creates bootstrap-like effects through parameter resampling:
\begin{equation}
\Theta_b = \{w_{ij} : m_{ij} = 1\} \subset \Theta
\end{equation}

The key insight is that $\Theta_b$ represents knowledge from an implicit data subset, making each dropout sample analogous to training on a bootstrap resample. Unlike explicit bootstrap methods that require retraining multiple models, dropout achieves similar uncertainty quantification through stochastic forward passes of a single trained network.

\subsection{Bayesian Interpretation}

From a Bayesian perspective, dropout approximates sampling from the posterior distribution over model parameters:
\begin{equation}
p(\Theta | \mathcal{D}) \approx \prod_{ij} \text{Bernoulli}(w_{ij}; 1-p)
\end{equation}

Each dropout mask corresponds to a different posterior sample, where the ``missing'' parameters represent uncertainty about which training examples are most relevant for the current prediction. This uncertainty naturally translates to predictive uncertainty, as different parameter subsets representing different implicit data views produce different predictions. The resulting variance in predictions across multiple forward passes provides a natural estimate of epistemic uncertainty.

\subsection{Empirical Observations and Limitations}

The bootstrap interpretation explains several empirical observations about dropout uncertainty. Higher dropout rates increase uncertainty because more parameters are dropped, removing more training knowledge and producing higher prediction variance. Uncertainty correlates with prediction difficulty since complex patterns requiring many parameters show higher variance when those parameters are randomly removed. Additionally, out-of-distribution detection works well because inputs unlike training data rely on broader parameter sets, making them more sensitive to dropout-induced parameter removal.

However, the dropout-as-bootstrap interpretation has important limitations. Parameter interdependence means that unlike independent data points, parameters have complex dependencies that aren't captured by simple Bernoulli sampling. The distinction between dropout during training versus inference means that inference-time dropout doesn't perfectly mirror the effect of training on different data subsets. Furthermore, the $1/(1-p)$ scaling during inference is a practical correction rather than a theoretically principled bootstrap operation.

\subsection{Practical Implications}

Understanding dropout as implicit bootstrapping provides valuable intuition for uncertainty calibration, explaining why dropout uncertainty estimates are often well-calibrated through their approximation of natural variability from different training data views. This perspective informs architecture design decisions, clarifying why deeper networks with more parameters benefit more from dropout due to having more parameters to ``bootstrap'' over. For hyperparameter tuning, this view suggests that dropout rates should be tuned based on dataset size and complexity, balancing information retention with uncertainty estimation quality.

This data-centric perspective bridges the gap between dropout's empirical success in uncertainty quantification and its theoretical foundations in approximate Bayesian inference, providing complementary understanding to traditional parameter-focused explanations of dropout regularization.

\end{document}